\documentclass[11pt,twocolumn]{article}
\usepackage{graphicx}
\usepackage{balance}

\usepackage[utf8]{inputenc}

\usepackage{algorithm}
\usepackage{mathtools}
\usepackage{xcolor}
\usepackage{listings}
\usepackage{amsmath,amssymb,amsfonts,amsthm}
\usepackage{authblk}

\newcommand{\tail}{\ensuremath{\texttt{Tail}}}

\newcounter{cntTheorem}
\newtheorem{theorem}{Theorem}[section]
\newtheorem{lemma}[cntTheorem]{Lemma}

\newcommand{\ram}{\ensuremath{\texttt{RAM}}{}}
\newcommand{\cbf}{\ensuremath{\texttt{CBF}}{}}

\newcommand{\dtw}{\ensuremath{\texttt{DTW}}{}}
\newcommand{\erp}{\ensuremath{\texttt{ERP}}{}}
\newcommand{\lcss}{\ensuremath{\texttt{LCSS}}{}}

\newcommand{\frechet}{\ensuremath{\texttt{DK}}{}}
\newcommand{\greedydk}{\ensuremath{\texttt{GDK}}{}}
\newcommand{\dist}{\ensuremath{\texttt{d}}{}}
\newcommand{\lbkeogh}{\ensuremath{\texttt{LB}_\texttt{Keogh}}{}}
\newcommand{\lbhausdorff}{\ensuremath{\texttt{LB}_{\Sigma\texttt{min}}}{}}
\newcommand{\lbbox}{\ensuremath{\texttt{LB}_\texttt{Box}}{}}
\newcommand{\normal}{\ensuremath{\Phi}}

\renewcommand{\geq}{\geqslant}
\renewcommand{\leq}{\leqslant}

\renewcommand{\phi}{\ensuremath{\varphi}}
\renewcommand{\epsilon}{\ensuremath{\varepsilon}}

\newcount\colveccount
\newcommand*\colvec[1]{
        \global\colveccount#1
        \begin{pmatrix}
        \colvecnext
}
\def\colvecnext#1{
        #1
        \global\advance\colveccount-1
        \ifnum\colveccount>0
                \\
                \expandafter\colvecnext
        \else
                \end{pmatrix}
        \fi
}

\newcommand{\PP}{\ensuremath{\mathbb{P}}}
\newcommand{\VAR}{\ensuremath{\mathbb{V}}}
\newcommand{\EE}{\ensuremath{\mathbb{E}}}
\newcommand{\NN}{\ensuremath{\mathbb{N}}}
\newcommand{\RR}{\ensuremath{\mathbb{R}}}

\begin{document}

\title{An Analytical Approach to Improving Time Warping on Multidimensional Time Series}

\author[1]{Jörg P. Bachmann}
\author[2]{Johann-Christoph Freytag}
\affil[1]{ \texttt{joerg.bachmann@informatik.hu-berlin.de}}
\affil[2]{ \texttt{freytag@informatik.hu-berlin.de}}
\affil[1,2]{Humboldt-Universität zu Berlin, Germany}

\date{\today}

\maketitle

\begin{abstract}
    Dynamic time warping (\dtw{}) is one of the most used distance functions to compare time series, e.\,g. in nearest neighbor classifiers.
    Yet, fast state of the art algorithms only compare 1-dimensional time series efficiently.
    One of these state of the art algorithms uses a lower bound (\lbkeogh) introduced by E. Keogh to prune \dtw{} computations.
    We introduce \lbbox{} as a canonical extension to \lbkeogh{} on multi-dimensional time series.
    We evaluate its performance conceptually and experimentally and show that an alternative to \lbbox{} is necessary for multi-dimensional time series.
    We also propose a new algorithm for the dog-keeper distance (\frechet) which is an alternative distance function to \dtw{} and show that it outperforms \dtw{} with \lbbox{} by more than one order of magnitude on multi-dimensional time series.
\end{abstract}

\section{Introduction}

Multimedia retrieval is a common application which requires finding similar objects to a query object.
We consider examples such as gesture recognition with modern virtual reality motion controllers, GPS tracking, speech recognition, and classification of handwritten letters where the objects are multi-dimensional time series.

In many cases, similarity search is performed using a distance function on the time series, where small distances imply similar time series.
A \emph{nearest neigbor query} to the query time series can be a $k$-nearest neighbor ($k$-NN) query or an $\epsilon$-nearest neighbor ($\epsilon$-NN) query.
A $k$-NN query retrieves the $k$ most similar time series.
An $\epsilon$-NN query retrieves all time series with a distance of at most $\epsilon$.

A common application of $k$-NN queries is solving classification tasks in machine learning \cite{kNN1,Clustering}.
For example, if all data objects (here time series) are labeled, a $k$-NN classifier basically assigns the label of a nearest neighbor to the unlabeled query object.

In our examples, the time series of the same classes, (e.\,g., same written character or same gestures) follow the same path in space, but have some temporal displacements.
Tracking the GPS coordinates of two cars driving the same route from A to B is an illustrative example.
We want these tracks (i.\,e. 2-dimensional time series) to be recognized as similar, although driving style, traffic lights, and traffic jams might result in temporal differences.
\emph{Time warping distance functions} such as dynamic time warping (\dtw)~\cite{DTWSakoe}, the dog-keeper distance (\frechet)~\cite{computingfrechet}, and the edit distance with real penalties (\erp)~\cite{ERP} respect this semantic requirement.
They map pairs of time series representing the same trajectory to small distances.

We are interested in fast algorithms for $k$-NN queries based on these time series distance functions.
Unfortunately, time warping distance functions usually have quadratic runtime~\cite{FrechetNoSubquadratic,DTWNoSubquadratic,EditNoSubquadratic}.

However, we are usually not interested in the exact distance values but in the set of the nearest neighbors to our query time series.
During such a nearest neighbor query, a common approach for improving the runtime is pruning as much distance computations as possible using lower bounds to the distance function:
If the lower bound exceeds a certain threshold already (e.\,g. the largest distance to our $k$ nearest neighbors found so far), the expensive distance also yields a larger value and thus, its computation can be skipped.

E.~Keogh proposed one of the state of the art algorithms for nearest neighbor queries with \dtw{} on 1-dimensional time series.
His main contribution is the lower bound \lbkeogh{}~\cite{Trillion,ExactIndexingDTW} pruning many expensive \dtw{} computations during a linear scan.

\subsection{Contributions}

The contributions of this paper are the following:
\begin{itemize}
    \item
        We introduce \lbbox{} as canonical extension to \lbkeogh{} for multi-dimensional time series in Section~\ref{sec:lbbox}, i.\,e. \lbbox{} is a lower bound to \dtw{} on multi-dimensional time series and equals \lbkeogh{} on 1-dimensional time series.
    \item
        In Section~\ref{sec:lbbox}, we theoretically show, that the canonical extension \lbbox{} suffers from an effect similar to the curse of dimensionality.
        In Section~\ref{sec:evaluation}, we confirm the theoretical results with experiments on two synthesized and several real world data sets.
    \item
        In Section~\ref{sec:dogkeeper}, we propose a new algorithm to compute the \frechet{} distance working faster than \dtw{} with \lbbox{} by more than one order of magnitude on multi-dimensional time series.
        The comparisons to \dtw{} with \lbbox{} in terms of accuracy and computation time are in Section~\ref{sec:evaluation}.
\end{itemize}

We first declare basic notation and preliminaries in Section~\ref{sec:notation}~and~\ref{sec:relatedwork}, respectively.

\subsection{Basic Notation}
\label{sec:notation}

We denote the natural numbers including zero with $\NN$ and the real numbers with $\RR$.
Elements of a $d$-dimensional vector $v\in\RR^d$ are accessed using subindices, i.\,e. $v_3$ is the third element of the vector.
Sequences (here also called time series) are usually written using capital letters, e.\,g. $S=(s_1,\cdots,s_n)$ is a sequence of length $n$.
Suppose $s_i\in\RR^d$, then $s_{i,j}$ denotes the $j$-th element of the $i$-th vector in the sequence $S$.
The projection to the $j$-th dimension is denoted via $S^j$, i.\,e. $S^j=(s_{1,j},\cdots,s_{n,j})$.
The Euclidean norm of a vector $v$ is denoted via $\|v\|_2$, thus $\|v-w\|_2$ denotes the Euclidean distance between $v$ and $w\in\RR^d$.
In general, we denote distance functions via $\dist$.

We denote $[l,u]\subset\RR$ as the 1-dimensional interval from $l\in\RR$ to $u\in\RR$.
We denote the cartesian product using the symbol $\bigotimes$, thus $\bigotimes_{1\leq j\leq d}[l_i,u_i]$ denotes the set of vectors $v$ with $l_j\leq v_j\leq u_j$ for $1\leq j\leq d$.

For a random variable $X$ over $\RR$, we denote the mean with $\EE[X]$, its variance using $\VAR[X]$, and the probability measure with $\PP$, i.\,e. $\PP[X<a]$ is the probability that the value of $X$ is less than $a\in\RR$.
We denote the standard normal distribution with $\normal_{0,1}$.

\section{Preliminaries}
\label{sec:relatedwork}

\subsection{Time Warping Distance Functions}

Dynamic time warping (\dtw) is a distance function on time series~\cite{DTWSakoe}.
Its benefit is a dynamic time alignment thus it is robust against time distortion or temporal displacements.
Distance functions on time series which are robust against time distortion are from hereon called \emph{time warping distance} functions.

For a formal definition, let $S=\left( s_1,\cdots,s_m \right)$ and $T=\left( t_1,\cdots,t_n \right)$ be two time series and $\dist$ a distance function for the elements of the time series.
\dtw{} is defined recursively:
\begin{align*}
    \dtw(S,()) =&\  \infty \quad \dtw((),T) = \infty \\
    \dtw((s),(t)) &= \dist(s,t)^2 \\
    \dtw(S,T) =& \ d(s_1,t_1)^2+\\
    &\quad\quad\min\begin{cases}
        \dtw( \tail(S), \tail(T) ) \\
        \dtw( S, \tail(T) ) \\
        \dtw( \tail(S), T )
    \end{cases}
\end{align*}
where $\tail(T)\coloneqq\left( t_2,\cdots,t_n \right)$.
The dog-keeper distance (\frechet) is similar to \dtw{}.
It only differs by taking the maximum distance along a warping path instead of the sum.
\begin{align*}
    &\frechet(S,()) =\ \infty \quad \frechet( (), T) = \infty \\
    &\frechet( (s), (t) ) = \dist(s,t) \\
    &\frechet(S,T) =\\
    &\quad \min\begin{cases}
        \max\left\{\dist(s_1,t_1),\frechet( \tail(S), \tail(T) ) \right\} \\
        \max\left\{\dist(s_1,t_1),\frechet( S, \tail(T) ) \right\} \\
        \max\left\{\dist(s_1,t_1),\frechet( \tail(S), T ) \right\}
    \end{cases}
\end{align*}
Further examples for time warping distance functions are \lcss{}~\cite{LCSS} and \erp{}~\cite{ERP}.

Well known algorithms computing \dtw{} and \frechet{} in quadratic time exist~\cite{DTWSakoe,computingfrechet}.
These algorithms first build the cross product of both time series $S$ and $T$ using the distance function on the elements.
The resulting \emph{distance matrix} consists of entries $\dist(s_i,t_j)^2$ where $\dist(s_1,t_1)^2$ is in the bottom left cell and $\dist(s_m,t_n)^2$ is in the top right cell.
After that, the algorithms compute the cheapest path from the bottom left to the top right cell by cell.
For each cell, they replace the value with a combination of that value and the smallest value of one of the possible predecessors: the left, the lower, and the left lower neighbor cell, yielding the \emph{warping matrix}.
Unfortunately, any algorithm computing any of these distance functions has quadratic runtime complexity in worst case~\cite{DTWNoSubquadratic,FrechetNoSubquadratic}.

\subsection{Sakoe Chiba Band}

The Sakoe Chiba band changes the semantic of \dtw{} by constraining the possible paths in the warping matrix to a diagonal band of a certain bandwidth~\cite{DTWSakoe}.
Thus, the warping of the time is constrained.
On the other hand, the computation time is decreased since a huge part of the distance matrix does not need to be considered.
Still, the runtime complexity remains quadratic because only a fixed ratio of the distance matrix is ``cut off''.

\subsection{Keoghs Lower Bound for DTW}
\label{sec:lbkeogh}

Cheap lower bounds are used to prune many complete \dtw{} computations and thus improve the overall computation time of a nearest neighbor search:
If the lower bound already exceeds a desired threshold, then the final result of the expensive \dtw{} distance function will be larger as well.
Keogh proposed one of the most successfull lower bounding functions to \dtw{}~\cite{ExactIndexingDTW,Trillion}.
His lower bound \lbkeogh{} depends on the Sakoe Chiba Band which constrains finding best time alignments to a maximum time distance of a certain band width.

\begin{figure}
    \centering
    \includegraphics[width=.49\linewidth]{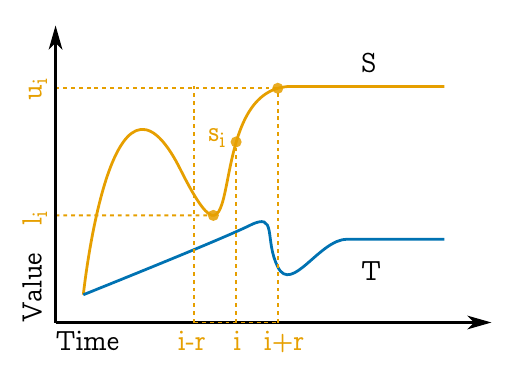}
    \includegraphics[width=.49\linewidth]{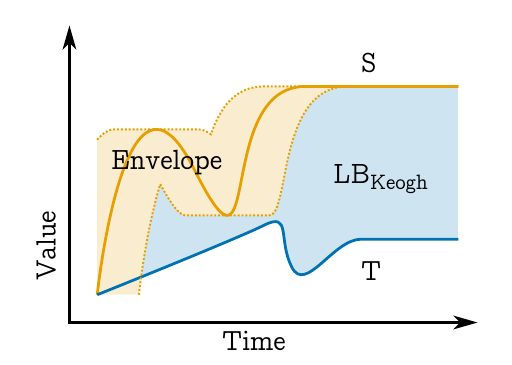}
    \caption{The envelope functions (computation sketched on the left) are used to compute the \lbkeogh{} lower bound (right).}
    \label{fig:lbkeogh}
\end{figure}

Basically, the idea of his lower bound is the following (cf. Figure~\ref{fig:lbkeogh} for an illustration):
Consider two time series $S=\left( s_1,\cdots,s_m \right)$ and $T=\left( t_1,\cdots,t_n \right)$ and map each $s_i\in S$ to the interval of all possible aligned values within the time range (i.\,e. Sakoe Chiba band) $i-r$ and $i+r$:
\begin{align*}
    l_i &\coloneqq \min\left\{ t_{i-r},\cdots,r_{i+r} \right\} \\
    u_i &\coloneqq \max\left\{ t_{i-r},\cdots,r_{i+r} \right\}
\end{align*}
Summing up the square distances of $s_i$ to each interval $[l_i,u_i]$ is a lower bound to \dtw{} since \dtw{} aligns each $s_i$ to at least one of the values within $[l_i, u_i]$, i.\,e.:
\begin{align*}
    \lbkeogh( S, T ) &\coloneqq \sum_{i=1}^{m} \dist\left( s_i, \left[ l_i, u_i \right] \right)^2 \leq \dtw(S,T)
\end{align*}
where
\begin{align*}
    \dist\left( s_i, \left[ l_i, u_i \right] \right) &\coloneqq \begin{cases}
        s_i-l_i & \text{iff } s_i < l_i \\
        u_i-s_i & \text{iff } u_i < s_i \\
        0 & \text{else}
    \end{cases}
\end{align*}

The computation of the intervals $[l_i,u_i]$ takes linear time when using the algorithm of Daniel Lemire~\cite{Lemire}.
Hence, it is obvious that the computation of \lbkeogh{} is linear in the length of the time series as well.

\section{Multi Dimensional Time Series}
\label{sec:lbbox}

Consider two time series $S=\left( s_1,\cdots,s_m \right)$ and $T=\left( t_1,\cdots,t_n \right)$ with $s_i,t_j\in\RR^{k}$, i.\,e. $S$ and $T$ are multi-dimensional time series, and let
\begin{align*}
    \dist( s_i,t_j ) &\coloneqq \|s_i-t_j\|_2
\end{align*}
be the Euclidean distance of the two vectors $s_i$ and $t_j$.
We extend the lower bound \lbkeogh{} of \dtw{} canonically using the interpretation presented in Section~\ref{sec:lbkeogh}:
For each $s_i$ we find the minimal bounding box to all values $\left\{ t_{i-r},\cdots,t_{i+r} \right\}$.
Summing up the square distances of $s_i$ to their bounding boxes is again a lower bound to \dtw.

For a formal definition, let
\begin{align*}
    B_i &\coloneqq \bigotimes_{1\leq j\leq k} \left[ l_{i,j},u_{i,j} \right] \\
    &= \bigotimes_{1\leq j\leq k}\left[ \min_{-r\leq \ell\leq r} t_{i+\ell,j}, \max_{-r\leq \ell\leq r} t_{i+\ell,j} \right]
\end{align*}
be the minimal axis aligned bounding box of $t_{i-r},\cdots,t_{i+r}$.
The vectors $l_i$ and $u_i$ are the lower left and upper right corners, respectively.
The distance of a vector to the bounding box is the distance to the nearest element within the bounding box:
\begin{align*}
    \dist\left( s_i,B_i \right)^2 &\coloneqq \min_{t\in B_i}\dist\left( s_i,B_i \right)^2 \\
    &= \min_{t\in B_i} \| s_i-t \|^2_2 \\
    &= \sum_{1\leq j\leq k} \dist\left( s_i, \left[ l_{i,j},u_{i,j} \right] \right)^2
\end{align*}
Considering the definition of \dtw, the following function is a lower bound to \dtw{}:
\begin{align*}
    \label{eq:lbbox}
    \dtw(S,T) &\geq \\
    \lbhausdorff(S,T) &\coloneqq \sum_{i=1}^m \min_{-r\leq\ell\leq r} d\left( s_i, t_{i+\ell} \right)^2 \\
\end{align*}
The distance function \lbhausdorff{} is hard to compute but can be estimated:
\begin{align*}
    \lbhausdorff(S,T) \geq \lbbox(S,T) &\coloneqq \sum_{i=1}^m \dist\left( s_i, B_i \right)^2 \\
    &\hspace{-3cm}= \sum_{i=1}^m \sum_{j=1}^k \dist\left( s_i, \left[ l_{i,j},u_{i,j} \right] \right)^2 \\
    &\hspace{-3cm}= \sum_{j=1}^k \lbkeogh\left( S^j, T^j \right) \addtocounter{equation}{1}\tag{\theequation}
\end{align*}
As it turns out, the estimation  is not only a canonical extension of the lower bound proposed by Keogh, but it also can be computed by using his proposed algorithm on the different dimensions of the time series.
Please note, that $\lbbox \equiv \lbkeogh$ holds especially for 1-dimensional time series.

\paragraph*{Curse Of Dimensionality}

\begin{figure}
    \centering
    \includegraphics[width=.85\linewidth]{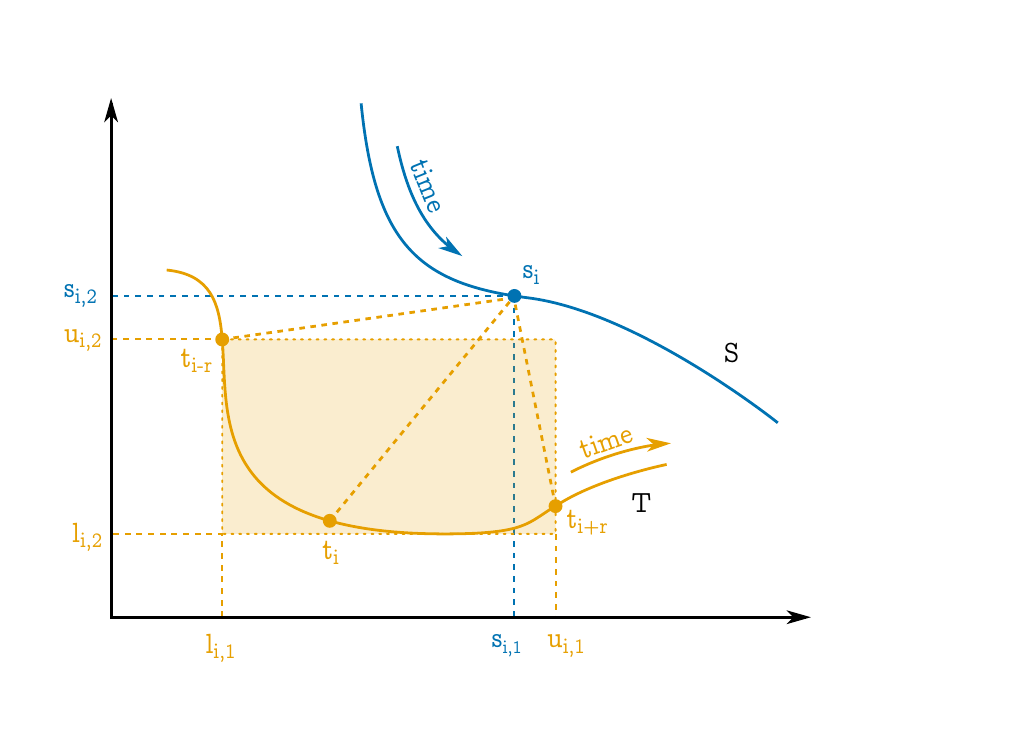}
    \caption{Two 2-dimensional time series.
        \dtw{} with Sakoe Chiba Band compares $q_i$ to at least one of $t_{i-r},\cdots,t_{i+r}$.
        \lbbox{} compares $q_i$ to the corresponding bounding box $[l_{i,1},u_{i,1}]\otimes[l_{i,2},u_{i,2}]$.}
    \label{fig:lbbox_issue}
\end{figure}

Consider the bounding box of a subsequence $\left( t_{i-r},\cdots,t_{i+r} \right)$ as illustrated in Figure~\ref{fig:lbbox_issue}.
For 2-dimensional time series, the following situation might happen:
The time series $T$ moves along the left and then along the bottom edge of the bounding box.
The query point $s_i$ however is at the top right vertex of the bounding box.
Hence, a perfect alignment of \dtw{} would still result in a distance
\begin{align*}
    \lbhausdorff(S,T) &= \cdots+\min_{-r\leq\ell\leq r} \dist\left( s_i,t_{i+\ell} \right)^2+\cdots > 0
\end{align*}
Simplifying the distance function to the bounding box results in
\begin{align*}
    \lbbox(S,T) &= \cdots+ \sum_{j=1}^k \dist\left( s_i, \left[ l_{i,j},u_{i,j} \right] \right)^2 + \cdots \\
    &= \cdots + 0 + \cdots = 0
\end{align*}
Thus, there is a clear divergence of \dtw{} and \lbbox.
With increasing dimensionality, there is more space for the time series to sneak past the query point, i.\,e. the probability for this situation to happen increases.
For this reason, we claim that the \emph{tightness} (i.\,e. the ratio $\frac{\lbbox}{\dtw}$) of the lower bound gets worse with increasing dimensionality.
This effect is similar to that of the \textit{Curse of Dimensionality} \cite{SearchingInMetricSpaces}, thus we still call it the same.

The goal of this section is the proof of Theorem~\ref{thm:curseofdimensionality} which claims the existence of the Curse of Dimensionality on \lbbox.
Since we want to assume as little as possible from the data sets, we assume that the time series consist of independent and identical distributed elements.
Section~\ref{sec:evaluation} confirms the theoretical results experimentally.
For proving Theorem~\ref{thm:curseofdimensionality}, we first need some technical lemmata.

\begin{lemma}
    \label{lem:helpit}
    Let $D_{\ell,j}$ be independent identically distributed random variables for $1\leq\ell\leq r$ and $1\leq j\leq k$.
    Furthermore, let $\mu_r\coloneqq\EE\left[ \min_{1\leq\ell\leq r} D_{\ell,j} \right] < \mu\coloneqq\EE\left[ D_{\ell,j} \right]$.
    Then,
    \begin{align*}
        \frac{\EE\left[ \sum_{j=1}^k \min_{1\leq\ell\leq r} D_{\ell,j} \right]}{\EE\left[\min_{1\leq\ell\leq r} \sum_{j=1}^k D_{\ell,j} \right]} \rightarrow \frac{\mu_r}{\mu} < 1 \ \text{ for }k\rightarrow\infty.
    \end{align*}
\end{lemma}
\begin{proof}
    Let $\mu\coloneqq\EE\left[ D_{\ell,j} \right]$ be the mean and $\sigma^2\coloneqq\VAR\left[ D_{\ell,j} \right]$ be the variance of the identically distributed variables $D_{\ell,j}$.
    The following inequation holds using calculation rules for expected values:
    \begin{align}
        \label{eq:helpit1}
        \EE\left[ \sum_{j=1}^k \min_{1\leq\ell\leq r} D_{\ell,j} \right] &= k\cdot\EE\left[ \min_{1\leq\ell\leq r} D_{\ell,1} \right] = k\cdot\mu_r
    \end{align}
    Let $\rho\coloneqq\EE\left[ \left|D_{\ell,j}^3\right| \right]$.
    Since we only consider finite data sets for our nearest neighbour queries, $\rho<\infty$ holds obviously.
    For a theoretical analysis on data sets with infinite many elements, this property is a necessary preliminary in order to estimate the denominator using the Berry-Esseen Theorem.
    Basically, the Berry-Esseen Theorem claims that the sum of random variables converges to a normal distribution for increasing number of summands.
    We also apply the well known Markov inequality to estimate the denominator:
    \begin{align*}
        \label{eq:helpit2}
        &\EE\left[ \min_{1\leq\ell\leq r} \sum_{j=1}^k D_{\ell,j} \right] \geq a\cdot \PP\left[ \min_{1\leq\ell\leq r}\sum_{j=1}^k D_{\ell,j} \geq a \right] \\
        &= a\cdot \PP\left[ \bigvee_{1\leq\ell\leq r} \sum_{j=1}^k D_{\ell,j}\geq a \right] \\
        &= a\cdot \left( 1-\PP\left[ \bigwedge_{1\leq\ell\leq r} \sum_{j=1}^k D_{\ell,j} < a \right] \right) \\
        &= a\cdot \left( 1-\PP\left[ \sum_{j=1}^{k} D_{1,j}<a \right]^r \right) \\
        &= a\cdot \left( 1-\PP\left[ \frac{\sum_{j=1}^k D_{1,j} - k\cdot\mu}{\sigma\sqrt{k}} < \frac{a-k\cdot\mu}{\sigma\sqrt{k}} \right]^r \right) \\
        &\geq a\cdot\left( 1-\left(\Phi_{0,1} \left( \frac{a-k\cdot\mu}{k\sqrt k} \right)+\frac{C\cdot\rho}{\sigma^3\cdot\sqrt{k}}\right)^r \right) \\
        &= k\cdot\mu\cdot\left( 1-\left(\frac{C\cdot\rho}{\sigma^3\cdot\sqrt{k}}\right)^r \right) \text{ for }a=k\cdot\mu \addtocounter{equation}{1}\tag{\theequation}
    \end{align*}
    where the first inequation is the Markov inequality and the second inequation is derived from the Berry-Esseen Theorem which also claims the existence of the constant $C$.

    Putting together Equation~(\ref{eq:helpit1}) and Inequation~(\ref{eq:helpit2}), we achieve the desired convergence:
    \begin{align*}
        \frac{\EE\left[ \sum_{j=1}^k \min_{1\leq\ell\leq r} D_{\ell,j} \right]}{\EE\left[ \min_{1\leq\ell\leq r} \sum_{j=1}^k D_{\ell,j} \right]} &\leq \frac{k\cdot\mu_r}{k\cdot\mu\cdot\left( 1-\left(\frac{C\cdot\rho}{\sigma^3\cdot\sqrt{k}}\right)^r \right)} \\
        &\hspace{-2cm}\longrightarrow \frac{\mu_r}{\mu} \text{ for }k\rightarrow\infty
    \end{align*}
\end{proof}

For $r=R\cdot n$ ($0<R\leq 1$), the following Lemma~\ref{lem:tozero} shows that the value $\frac{\mu_r}{\mu}$ from Lemma~\ref{lem:helpit} converges to $0$ for $n\rightarrow\infty$.

\begin{lemma}
    \label{lem:tozero}
    Let $X_{i}\geq 0$ be independent identically distributed random variables for $1\leq i\leq n$ such that $\PP\left[ X_i\geq x \right]<1$ for each $x>0$.
    Then, $\EE\left[ \min_{1\leq i\leq n}X_i \right] \longrightarrow 0$ for $n\rightarrow\infty$.
\end{lemma}
\begin{proof}
    The expected value can be calculated as follows:
    \begin{align*}
        \EE\left[ \min_{1\leq i\leq n} X_i \right] &= \int_0^\infty \PP\left[ \min_{1\leq i\leq n} X_i \geq x \right] dx \\
        &= \int_0^\infty \PP\left[ X_1\geq x\wedge \cdots \wedge X_n\geq x \right] dx \\
        &= \int_0^\infty \PP\left[ X_i\geq x \right]^n dx
    \end{align*}
    Since $\PP\left[ X_i\geq x \right]$ is Lebesgue measurable and $\PP\left[ X_i\geq x \right]^n\longrightarrow 0$ for $x>0$, the last integral converges to zero for $n\rightarrow\infty$.
\end{proof}

Now, we prove that \lbbox{} suffers from the curse of dimensionality using Lemma~\ref{lem:helpit} and \ref{lem:tozero}, i.\,e. the tightness of the lower bound \lbbox{} to \dtw{} gets worse for increasing dimensionality.
\begin{theorem}
    \label{thm:curseofdimensionality}
    Let $S$ and $T$ be two time series in $\RR^k$ with length $n$.
    Then
    \begin{align*}
        \frac{\EE\left[\lbbox(S,T)\right]}{\EE\left[\dtw(S,T)\right]} \longrightarrow 0
    \end{align*}
    for $k,n\rightarrow\infty$.
\end{theorem}
\begin{proof}
    Since we do not know anything about the time series, we assume that their elements are independent identically distributed variables.
    For simplicity, we even assume that the distance between any two elements is a random variable, i.\,e.
    \begin{align*}
        D_{i,j} &\coloneqq \dist\left( s_i, t_j \right)^2
    \end{align*}
    are independent identically distributed random variables.
    With a Sakoe Chiba band of width $r=R\cdot n$ ($0<R\leq 1$ constant) we get
    \begin{align*}
        \frac{\EE\left[\lbbox(S,T)\right]}{\EE\left[\dtw(S,T)\right]} &\leq \frac{\EE\left[ \sum_{i=1}^n \sum_{j=1}^k \dist\left( s_i, \left[ l_{i,j},r_{i,j} \right] \right)^2 \right]}{\EE\left[\lbhausdorff(S,T)\right]} \\
        &\hspace{-2cm}\leq \frac{\EE\left[ \sum_{i=1}^n \sum_{j=1}^k \min_{-r\leq\ell\leq r} D_{i+\ell,j} \right]}{\EE\left[ \sum_{i=1}^n\min_{-r\leq\ell\leq r} \sum_{j=1}^k D_{i+\ell,j} \right]} \\
        &\hspace{-2cm}= \frac{n\cdot\EE\left[ \sum_{j=1}^k \min_{1\leq\ell\leq 2r+1} D_{\ell,j} \right]}{n\cdot\EE\left[ \min_{1\leq\ell\leq 2r+1} \sum_{j=1}^k D_{\ell,j} \right]} \\
        &\hspace{-2cm}\longrightarrow \frac{\mu_r}{\mu} \quad \text{ for } k\rightarrow \infty \\
        &\hspace{-2cm}\longrightarrow 0 \quad \text{ for } n\rightarrow \infty
    \end{align*}
\end{proof}

\section{Dog-Keeper Distance}
\label{sec:dogkeeper}

Practical implementations of \dtw{} speed up the computation process by stopping early when a \emph{promised lower bound} already exceeds a certain threshold.
For example, processing the warping matrix while computing \dtw{}, the minimum value of a column or row is a lower bound to the final distance value.
The lower bound proposed by Keogh promises even tighter values to the final distance value~\cite{Trillion}.

Since \dtw{} sums up values along a warping path through the distance matrix, the values in the early computation time exceed a certain threshold less probable than the values at a later computation time.
This observation does not hold when computing the \frechet{} distance~\cite{WDK17}.
This insight yields the idea that the matrix filled out during computation of the \frechet{} distance might be very sparse.

Therefore, our approach to improving the computation time of the \frechet{} distance is to compute the distance matrix using a \textit{sparse matrix} algorithm (cf. Section~\ref{sec:sparsedk}).
However, to avoid computing most of the cells of this matrix, we need a low threshold beforehand.
Such a threshold is found using a cheap upper bound to the \frechet{} distance.
Specifically, we propose a greedy algorithm to the time warping alignment problem (cf. \ref{sec:greedydk}).

To explain the algorithms in detail, consider two time series $S=(s_1,\cdots,s_m)$ and $T=(t_1,\cdots,t_n)$.

\subsection{Greedy Dog-Keeper}
\label{sec:greedydk}

The greedy dog-keeper distance (\greedydk) starts by aligning $s_1$ and $t_1$.
It then successively steps to one of the next pairs aligning $(s_{i+1}, t_j)$, $(s_i, t_{j+1})$, or $(s_{i+1}, t_{j+1})$ with the lowest distance.
When it reaches the alignment of $s_m$ to $t_n$, the maximum of the distances along the choosen path yields the final distance.
Algorithm~\ref{alg:greedydk} provides the pseudo code for the algorithm.

\begin{algorithm}
    \caption{Greedy dog-keeper distance}
    \begin{lstlisting}
Algorithm: greedydogkeeper
Input: time series $S,T$ of length $m,n$ resp.; threshold $\epsilon$
Output: upper bound $d$

let $(i,j)=(1,1)$
let $g=\dist\left( s_i,t_j \right)$
while $i\neq m$ and $j\neq n$
 // non defined distances yield $\infty$
 $(i,j)=\arg\min\begin{cases}
     \dist(s_{i+1},t_j) \\
     \dist(s_{i+1},t_{j+1}) \\
     \dist(s_i,t_{j+1})
 \end{cases}$
 $g=\max\left\{ g,\dist(i,j)\right\}$
 if $g>\epsilon$ then return $g$
return $g$
    \end{lstlisting}
    \label{alg:greedydk}
\end{algorithm}

So far, the \greedydk{} distance matches whole sequences against each other.
In order to support sub sequence matching, we alter the algorithm by first finding the best match of $s_1$ to any $t_i$ (name it $t_{left}$).
Then, we run \greedydk{} starting at $t_{left}$ and stop the computation as soon, as $s_m$ is aligned to one of the $t_i$.
For details, refer to Algorithm~\ref{alg:subgreedydk}.

\begin{algorithm}
    \caption{Greedy dog-keeper for Sub Sequence Search}
    \begin{lstlisting}
Algorithm: greedydogkeeper
Input: time series $S$ and $T$ of length $m$ and $n$ resp., threshold $\epsilon$
Output: upper bound $d$

let $(i,j)=(1,\arg\min\left\{ \dist(s_1,t_j) \right\})$
let $g=\dist\left( s_i,t_j \right)$
while $i\neq m$
 // non defined distances yield $\infty$
 $(i,j)=\arg\min\begin{cases}
     \dist(s_{i+1},t_j)\\
     \dist(s_{i+1},t_{j+1})\\
     \dist(s_i,t_{j+1})
 \end{cases}$
 $g=\max\left\{ g,\dist(i,j)\right\}$
 if $g>\epsilon$ then return $g$
return $g$
    \end{lstlisting}
    \label{alg:subgreedydk}
\end{algorithm}

Both algorithms have linear complexity, since the while loop runs at maximum $m+n$ times

\subsection{Sparse dog-keeper distance}
\label{sec:sparsedk}

The sparse dog-keeper algorithm essentially works the same as the original algorithm, except that it only visits those neighbor cells of the distance matrix having a value not larger than a given threshold.
Algorithm~\ref{alg:sparsedk} provides the pseudo code for the algorithm:
Similar to the original algorithm, we compute the (sparse) warping matrix column by column (cf.~Line~10).
The variables $I$ and $J$ store the indices of cells to visit in the current and next column, respectively.
If the value at the current cell of the matrix is not larger than the threshold (cf.~Line~13), we also need to visit the right, upper right, and upper cells (cf.~Lines~15~and~16).
The actual values within the columns are stored in $D$ (current column) and $E$ (previous column).
After we finished the computation of a column, we prepare the variables to enter the next column (cf.~Lines~19~to~26).

\begin{algorithm}
    \caption{Sparse dog-keeper distance}
    \begin{lstlisting}
Algorithm: sparsedogkeeper
Input: time series $S$ and $T$ of length $m$ and $n$ resp., threshold $\epsilon$, boolean SUB
Output: distance $d$

$I=\left\{ 1 \right\}$
$J=\emptyset$
$D=(\infty,\cdots)$
$E=(\infty,\cdots)$

for $k=1$ to $n$
  if SUB $d=\infty$
  for $i\in I$
    if $\dist\left( S_i, T_k \right) \leq \epsilon$
    $D_i=\max\begin{cases}
        \dist\left( S_i,T_k \right)\\
        \min\left\{ D_{i-1}, E_i, E_{i-1} \right\}
    \end{cases}$
      $I=I\cup\left\{ i+1 \right\}$
      $J=J\cup\left\{ i, i+1 \right\}$
    else
      $I=I\setminus\left\{ i \right\}$
  if SUB $d=\min\left\{ d, D_m \right\}$
  $E=D$
  // reset $D$ to $(\infty,\cdots)$
  for $i\in I$
    $D_i=\infty$
  if SUB $I=J\cup\left\{ 1 \right\}$
  else $I=J$
  $J=\emptyset$
  if not SUB and $I$ is empty
    return $\epsilon$

if SUB return $d$
else return $D_m$
    \end{lstlisting}
    \label{alg:sparsedk}
\end{algorithm}

The algorithm performs subsequence search iff passing \texttt{true} for the parameter \texttt{SUB}.
It differs to the whole matching version by considering each position of the super sequence as possible start of a match (cf.~Line~24) and by considering each position as possible end of a match (cf.~Line~19).

\subsection{Nearest Neighbour Search}
\label{sec:sparsedkNN}

Our nearest neighbor search algorithm requires finding a good upper bound.
Hence, we first loop over all time series and find the lowest upper bound using the \greedydk{} distance.
Then, we perform a common nearest neighbor search by scanning the time series.

\section{Evaluation}
\label{sec:evaluation}

In Section~\ref{sec:evaluateCoD}, we evaluate Theorem~\ref{thm:curseofdimensionality} (which claims the existence of the curse of dimensionality for \lbbox) experimentally on synthetic data sets which have a parameter for setting the dimensionality \cite{TSGenerator}.
We confirm our observations on real world data sets.

We compare \dtw{} (with \lbbox{} as lower bound) against our implementation of the \frechet{} distance in terms of runtime in Section~\ref{sec:evaluateTime}.
We focus on the subsequence matching algorithms since this is the more challenging task.

The retrieval quality is of high importance when considering an alternative distance function for nearest neighbor queries.
Therefore, we close the evaluation by comparing the accuracy of the \frechet{} and \dtw{} distance functions on nearest neighbor tasks in Section~\ref{sec:evaluateAcc}.

\paragraph*{Synthetic Data Sets}

Due to space limitations, we can not show results for all parameters of the data set generators.
The following set of parameters yield the most interesting results:
If not mentioned otherwise explicitely, we generate data with distortion $25$, radius $50$, $50$ distinguish classes, and $2$ representatives per class using the \ram{} generator and $27$ distinguish classes and $3$ representatives per class using the \cbf{} generator.

\subsection{Curse Of Dimensionality}
\label{sec:evaluateCoD}

Theorem~\ref{thm:curseofdimensionality} claims that the tightness of the lower bound function $\lbbox$ to $\dtw$ gets worse with increasing dimensionality.
We evaluate the theorem experimentally on data sets generated by the two synthesizers \cbf{} and \ram{} \cite{TSGenerator}.
Figure~\ref{fig:examples_ram} illustrates examples for time series generated by the \ram{} and \cbf{} synthesizers.

\begin{figure}
    \centering
    \includegraphics[width=.49\linewidth]{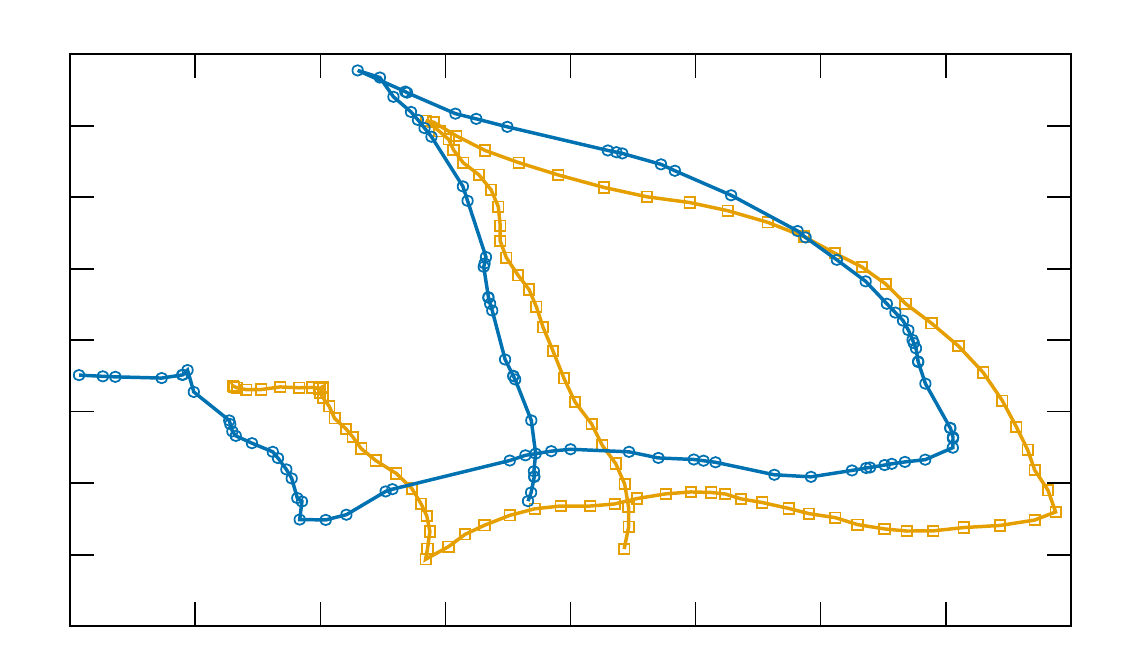}
    \includegraphics[width=.49\linewidth]{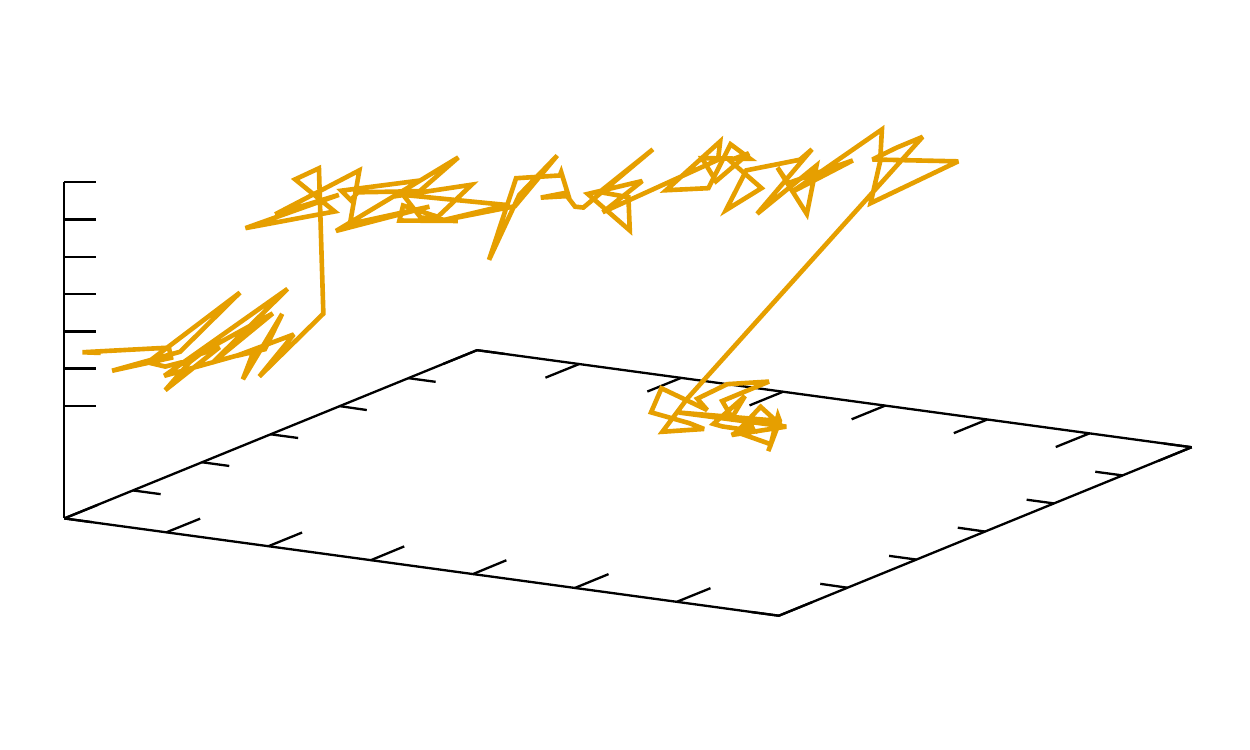}
    \caption{Two 2-dimensional examples for the \ram{} data set (left) and one example for the \cbf{} data set (right).}
    \label{fig:examples_ram}
\end{figure}

\paragraph*{Implementation}
Our implementation of \lbbox{} is based on \lbkeogh{} from the UCR Suite \cite{Trillion}.
We skipped the normalization on some data sets to improve the runtime and accuracy of \lbbox{}.
While adapting the UCR Suite to work on multi-dimensional time series, we checked that the runtime remained stable for 1-dimensional time series.
Thus, we made sure that all of our results for runtime comparison are not implementation dependent.
We ran all experiments on one core of an Intel(R) Xeon(R) CPU E5-2620 0 @ 2.00GHz with 24GB Memory.

\paragraph*{Tightness}
Figure~\ref{fig:tightness_ram_len_dim} demonstrates that the tightness of \lbbox{} to \dtw{} is decreasing down to zero for increasing length, which confirms Theorem~\ref{thm:curseofdimensionality}.
The theorem also claims that the tightness drops to a constant value for increasing dimensionality but constant length.
With Figure~\ref{fig:tightness_ram_len_dim} we confirm this claim for the \ram{} data sets and show even more that the tightness converges (drops down) already for moderate dimensionality (e.\,g. 3 dimensional time series).
We could ovbserve the same results for the \cbf{} data sets (cf. Figure~\ref{fig:tightness_cbf_len_dim}) although the dimensionality has greater impact on the tightness while the length has a smaller impact.

\paragraph*{Pruning power}
We show the effects on the pruning power in the second heat maps of both Figures~\ref{fig:tightness_ram_len_dim}~and~\ref{fig:tightness_cbf_len_dim}.
Their black lines (lines of same values within each heat map) look rather similar to the black lines in their corresponding heat maps on the left side.
They confirm the correlation between the dropping of the tightness of \lbbox{} and the dropping of the pruning power within a nearest neighbor search.
Still, the dimensionality has a larger impact to the pruning power than to the tightness of the time series.

\begin{figure}
    \centering
    \includegraphics[width=.49\linewidth]{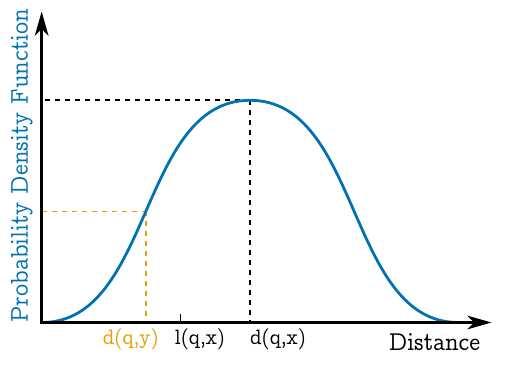}
    \includegraphics[width=.49\linewidth]{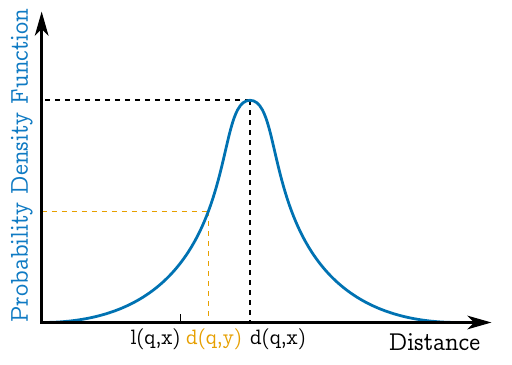}
    \caption{Two data sets (left: low dimensional; right: high dimensional), a query $q$, a nearest neighbor $y$ and a candidate $x$.
        If the lower bound $l(q,x)\geq d(q,y)$, then $x$ can be pruned without computation of $d(q,x)$.
    }
    \label{fig:curse_const_tightness}
\end{figure}

The reason is the curse of dimensionality on retrieval tasks \cite{SearchingInMetricSpaces}.
It basically says that the variance of distance values between random elements of a data set decreases with increasing dimensionality.
Thus, even if the tightness of a lower bound such as \lbbox{} equals on two distinct data sets, pruning is less probable in the higher dimensional data set.
Figure~\ref{fig:curse_const_tightness} illustrates an example where $q$ is the query and $y$ the nearest neighbor found so far during the search.
If the lower bound $l(q,x)$ is larger than $d(q,y)$ then we can prune the next element $x$.
Assume $l(q,x) = \alpha\cdot d(q,x)$ where $0\leq\alpha\leq 1$ is equal in the low dimensional and the high dimensional data set.
Since $d(q,x)$ and $d(q,y)$ converge on data sets with increasing dimensionality (this is the curse of dimensionality), $l(q,x)<d(q,y)\approx d(q,x)$ is more probable in higher dimensionality.
The lower bound can not be used for pruning in these cases.
Please note, that this insight holds for any lower bound driven approach.

Both of the Figures~\ref{fig:tightness_ram_len_dim} and \ref{fig:tightness_cbf_len_dim} show that \lbbox{} loses its benefit already for moderate dimensionality.

\begin{figure}
    \centering
    \includegraphics[width=.49\linewidth]{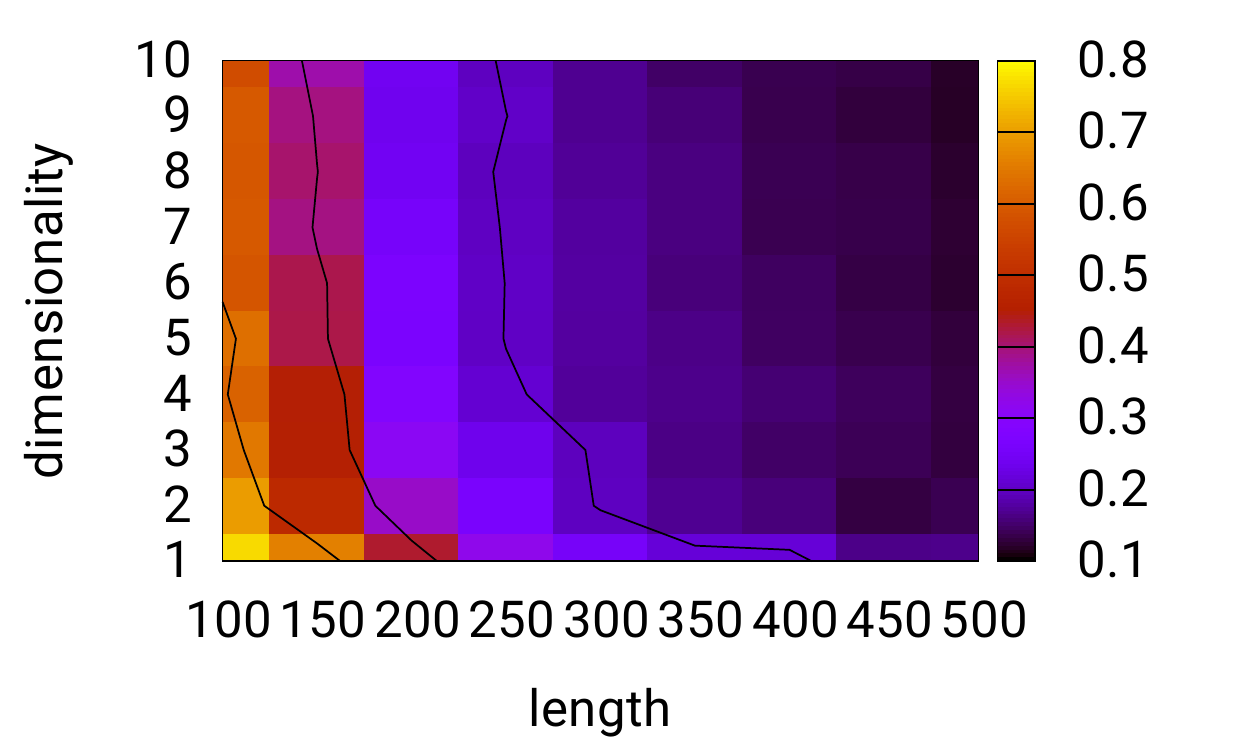}
    \includegraphics[width=.49\linewidth]{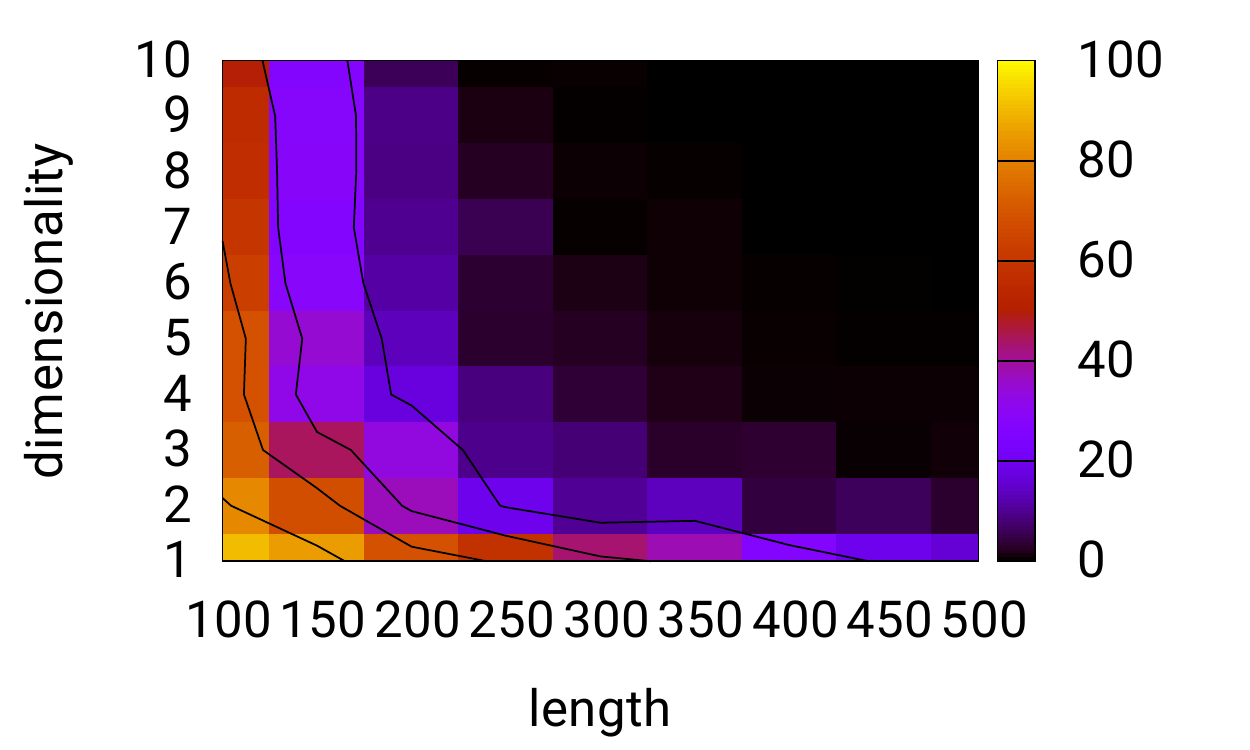}
    \caption{Heat maps presenting the average tightness of \lbbox{} (left) and the pruning power of \lbbox{} in percent (right)
    on the \ram{} data set.}
    \label{fig:tightness_ram_len_dim}
\end{figure}

\begin{figure}
    \centering
    \includegraphics[width=.49\linewidth]{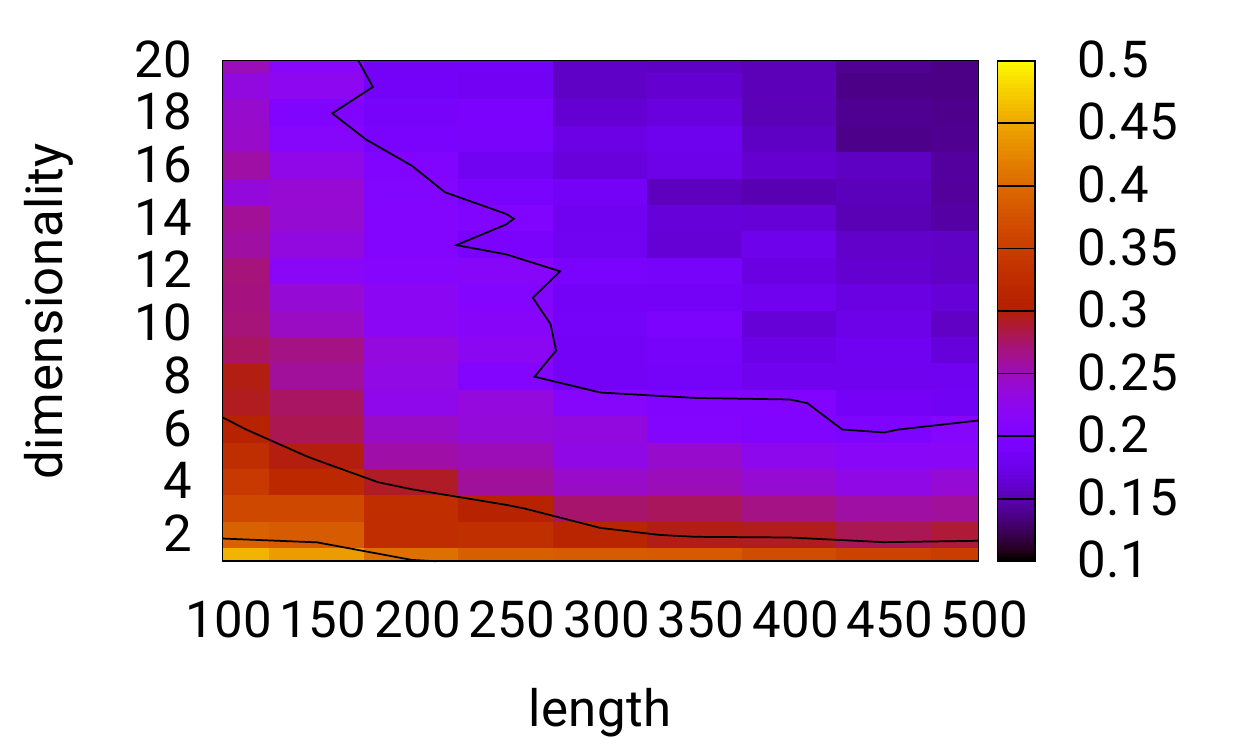}
    \includegraphics[width=.49\linewidth]{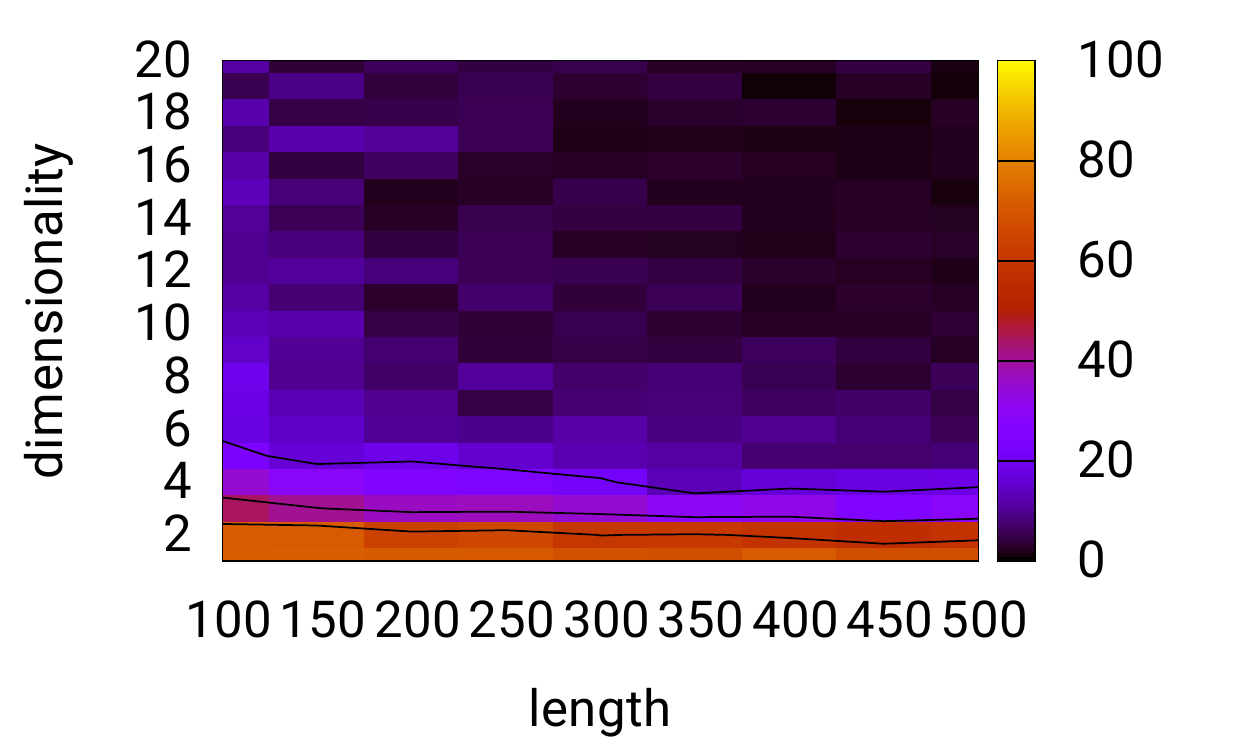}
    \caption{Heat maps presenting the average tightness of \lbbox{} (left) and the pruning power of \lbbox{} in percent (right)
    on the \cbf{} data set.
    }
    \label{fig:tightness_cbf_len_dim}
\end{figure}

\subsection{Computation Time}
\label{sec:evaluateTime}

\begin{figure}
    \centering
    \includegraphics[width=.49\linewidth]{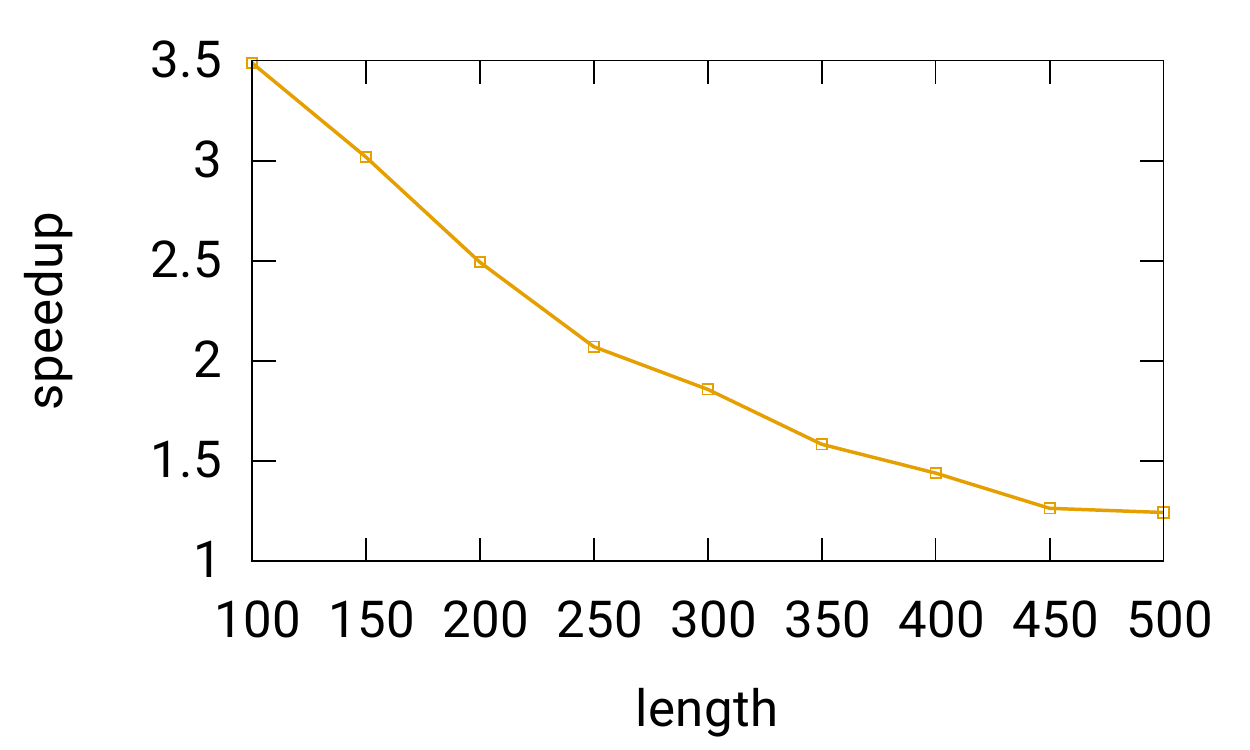}
    \includegraphics[width=.49\linewidth]{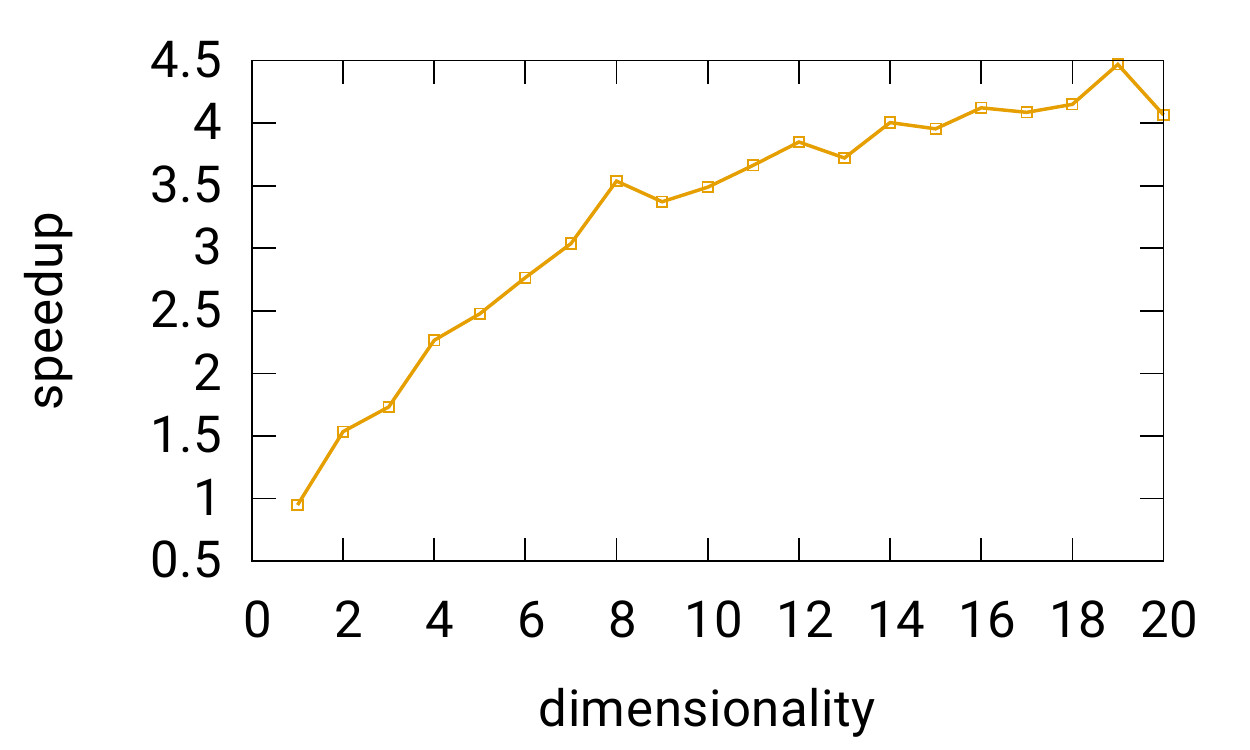}
    \caption{Speedup of \frechet{} to \dtw{} with \lbbox{} on the \cbf{} data set. Dimensionality: 10 (left); Length: 100 (right).}
    \label{fig:speedup_cbf_len_dim}
\end{figure}

\begin{figure}
    \centering
    \includegraphics[width=.49\linewidth]{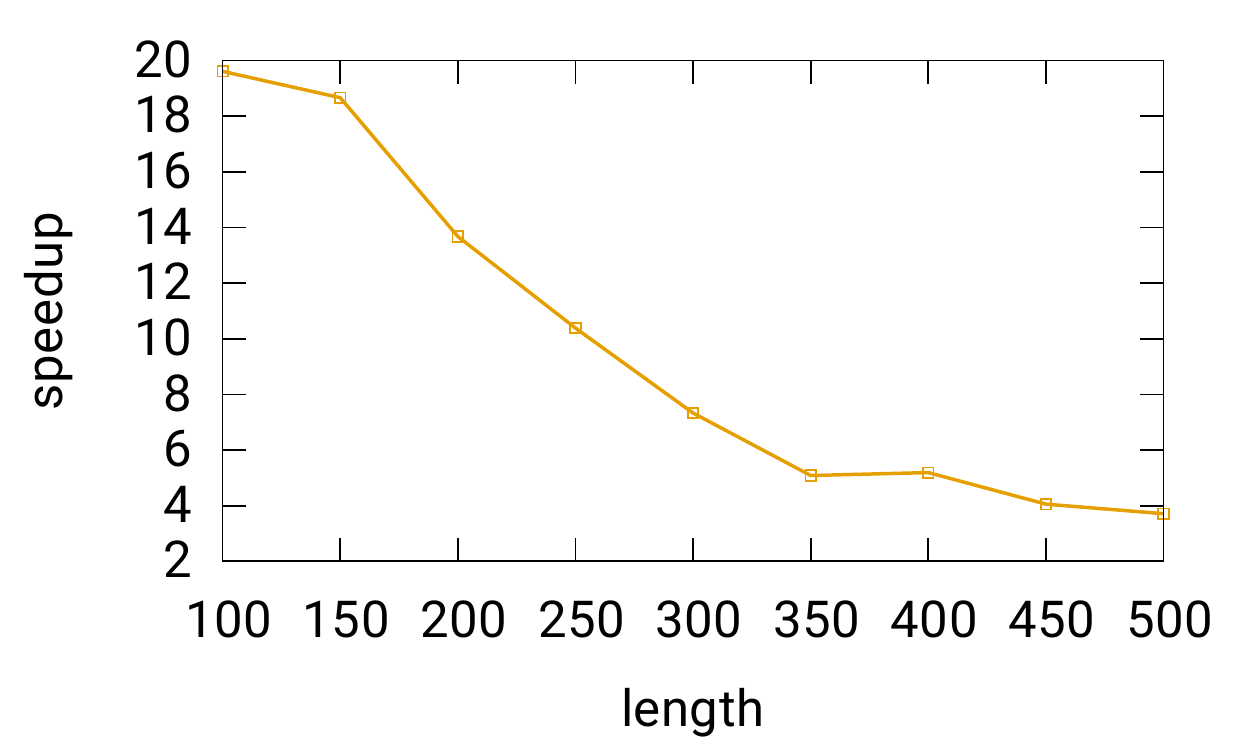}
    \includegraphics[width=.49\linewidth]{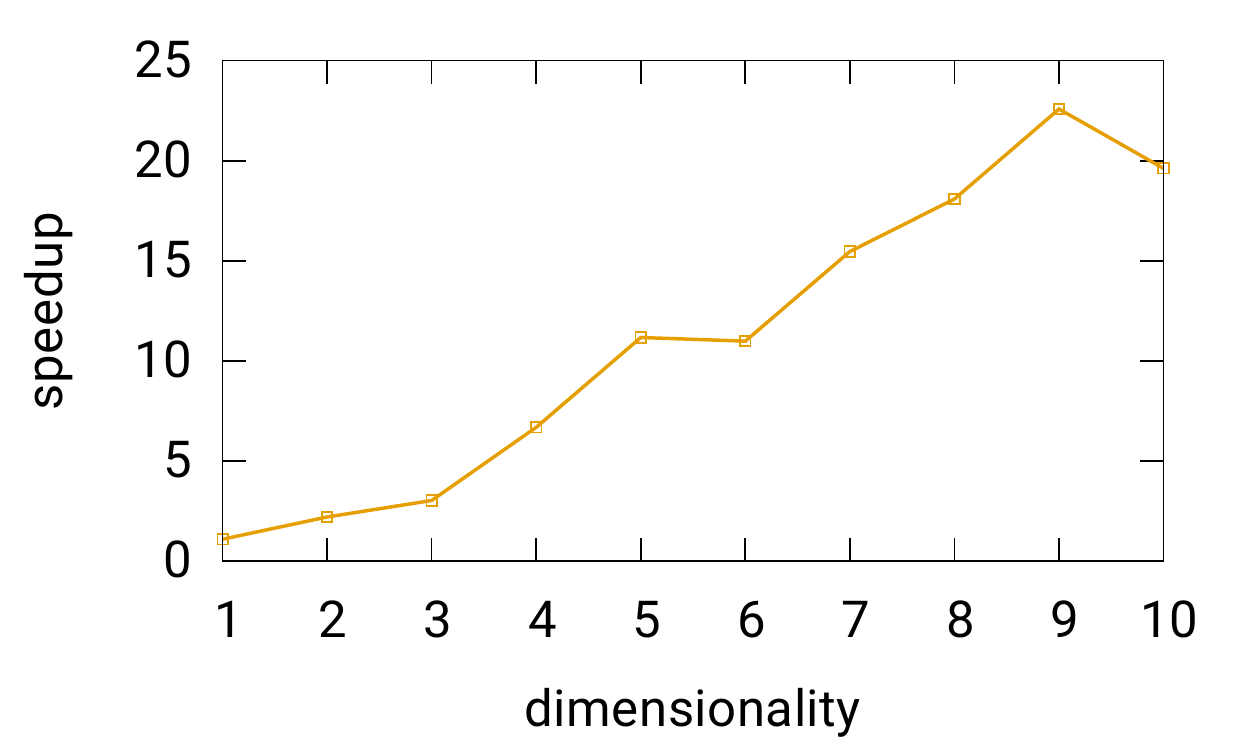}
    \caption{Speedup of \frechet{} to \dtw{} with \lbbox{} on the \ram{} data set. Dimensionality: 10 (left); Length: 100 (right).}
    \label{fig:speedup_ram_len_dim}
\end{figure}

\begin{figure}
    \centering
    \includegraphics[width=.49\linewidth]{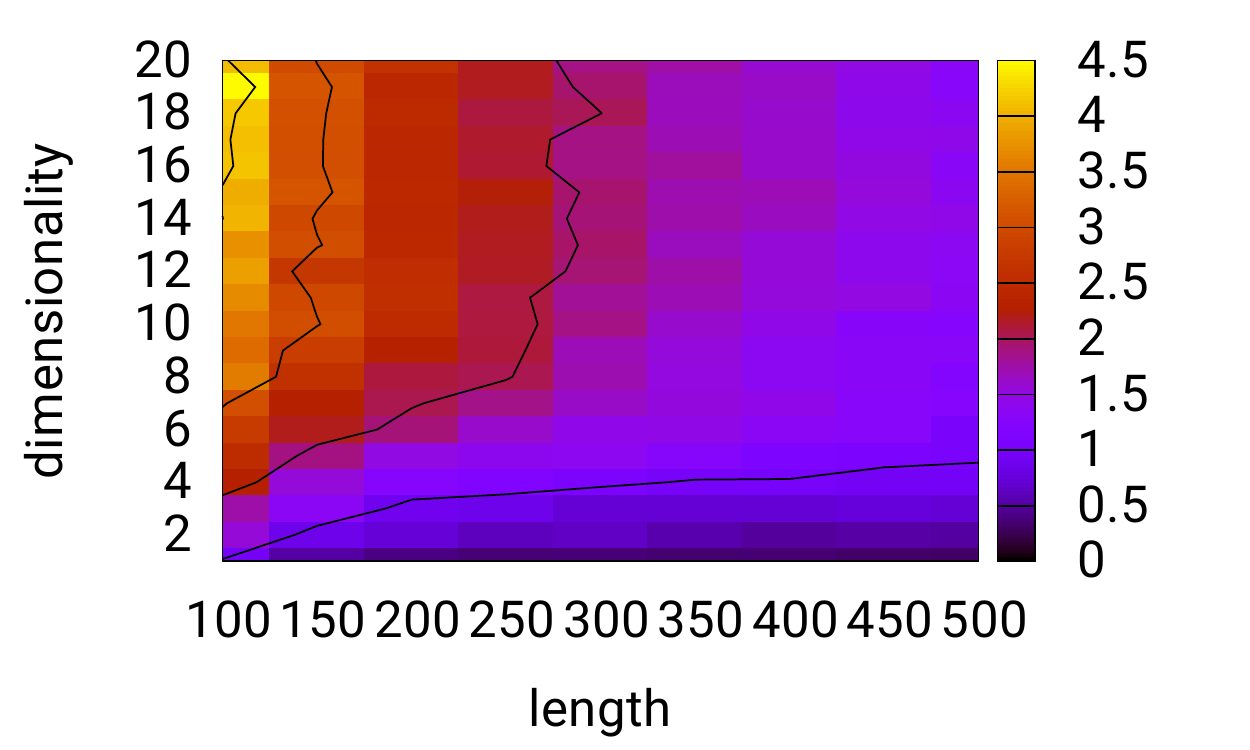}
    \includegraphics[width=.49\linewidth]{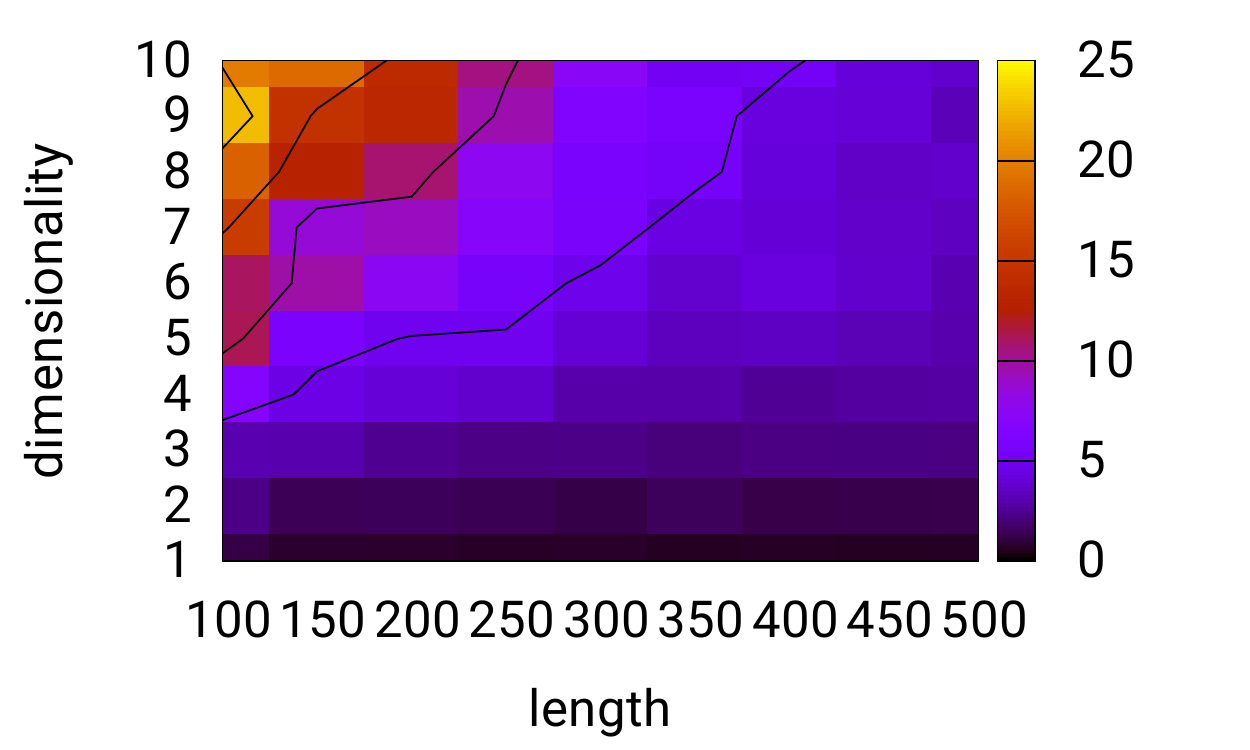}
    \caption{Speedup of \frechet{} to \dtw{} with \lbbox{} on the \cbf{} data set (left) and the \ram{} data set (right). Dimensionality: 10 (left); Length: 100 (right).}
    \label{fig:speedup_len_dim}
\end{figure}

For the evaluation of the runtime of our \frechet{} algorithm, we consider the runtime of \dtw{} with \lbbox{} as the base line.
In Figures~\ref{fig:speedup_cbf_len_dim} and \ref{fig:speedup_ram_len_dim} we show the speedup of our \frechet{} implementation against \dtw{} with \lbbox{} on the \ram{} and \cbf{} data sets, respectively.
The speedup decreases for longer time series while increasing with growing dimensionality.
Figure~\ref{fig:speedup_len_dim} confirms this observation on an even larger parameter set.
Thus, our \frechet{} implementation outperforms \dtw{} with \lbbox{} on rather short and multi-dimensional time series.

To ensure that our results do not depend on the generated synthetical data sets, we repeated the same experiments on the following real world data sets from the UCI Machine Learning Repository \cite{UCIDatasets}:
Character Trajectories (CT), Activity Recognition system based on Multisensor data fusion (AReM) \cite{AReM}, EMG Physical Action (EMGPA), Australian Sign Language 2 (ASL) \cite{ASL}, Arabic Spoken Digits (ASD), and Vicon Physical Action (VICON).
We also used the ECG data proposed by Keogh \cite{Trillion} for querying in one very long time series.
Columns~1~and~2 of Table~\ref{tbl:realdatasets} show that the speedup increases with growing dimensionality.
Thus, we demonstrate that our implementation of the \frechet{} distance outperforms \dtw{} with \lbbox{} in terms of computation time, especially on multi-dimensional time series.

\begin{table}
    \caption{Speedup and Accuracy of \dtw{} with \lbbox{} and \frechet{} on real world data sets.}
    \begin{center}
    \begin{tabular}{c|cccc}
        data set & dim. & speedup & acc. \dtw & acc. \frechet \\
        \hline
        ECG & 1 & $0.05$ & - & - \\
        CT & 2 & $4.5$ & $0.94$ & $ 0.93$ \\
        AReM \footnote{No Z-Normalization} & 6 & $1.2$ & $0.81$ & $0.75$ \\
        EMGPA & 8 & $31$ & $0.21$ & $0.26$ \\
        ASD & 13 & $45$ & $0.98$ & $0.96$ \\
        ASL & 22 & $33$ & $0.85$ & $0.88$ \\
        VICON & 27 & $62$ & $0.12$ & $0.09$
    \end{tabular}
    \end{center}
    \label{tbl:realdatasets}
\end{table}

\subsection{Retrieval Quality:}
\label{sec:evaluateAcc}

Figures~\ref{fig:accuracy_cbf_len_dim} and \ref{fig:accuracy_ram_len_dim} reveal that the accuracy of both distance functions decreases on \cbf{} while increasing on \ram{} with growing dimensionality.
While \frechet{} outperforms \dtw{} on the \cbf{} data set it looses on the \ram{} data set.

\begin{figure}
    \centering
    \includegraphics[width=.49\linewidth]{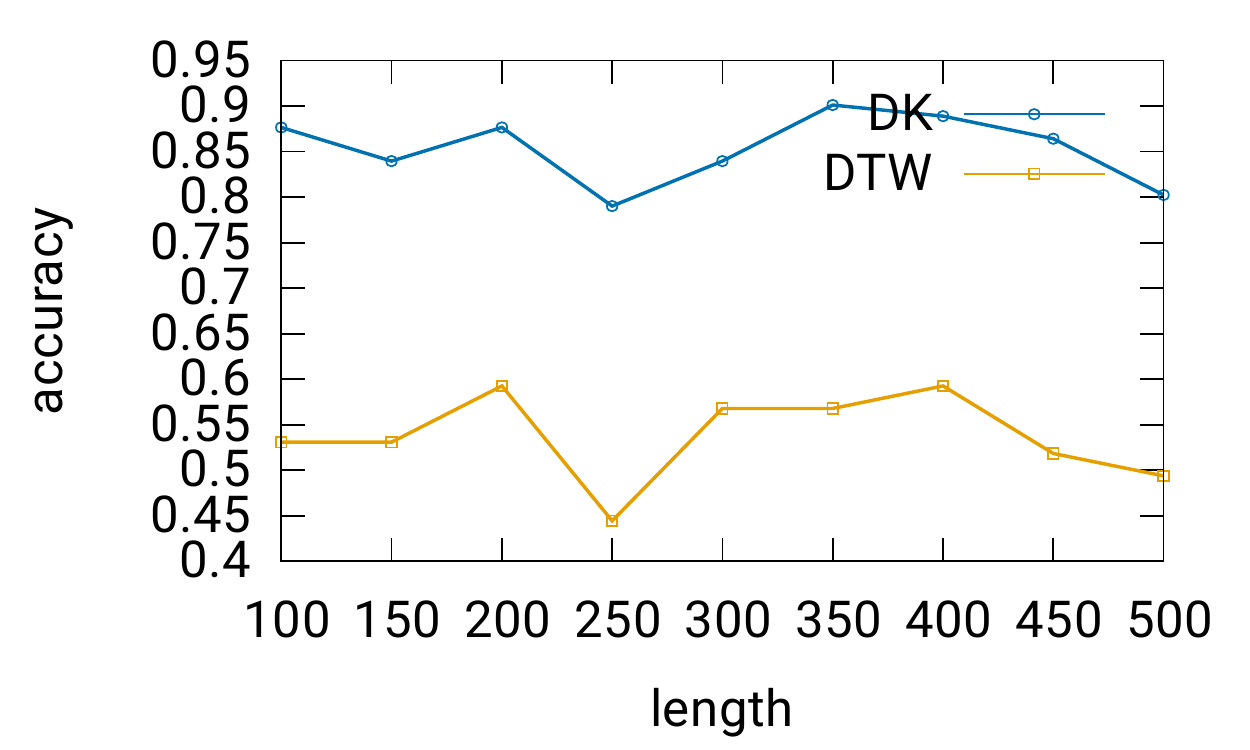}
    \includegraphics[width=.49\linewidth]{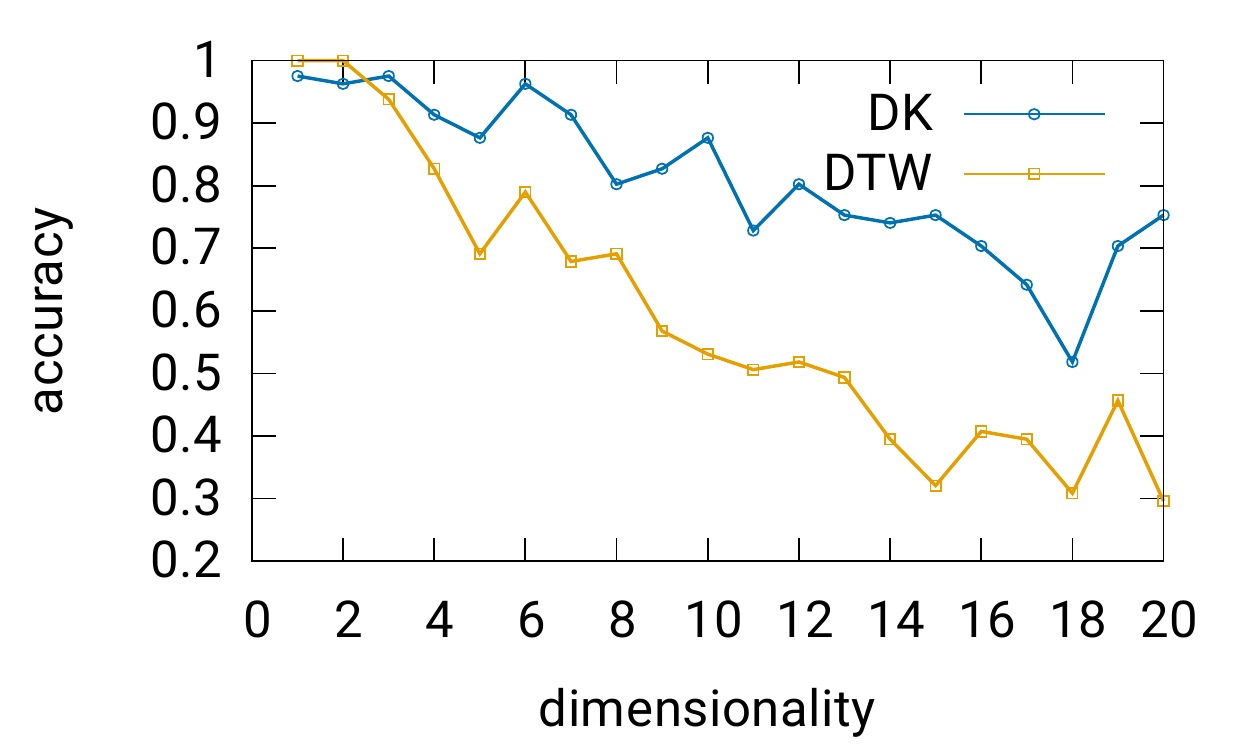}
    \caption{
        Accuracy of \frechet{} and \dtw{} with \lbbox{} on the \cbf{} data set. Dimensionality: 10 (left); Length: 100 (right).
    }
    \label{fig:accuracy_cbf_len_dim}
\end{figure}

\begin{figure}
    \centering
    \includegraphics[width=.49\linewidth]{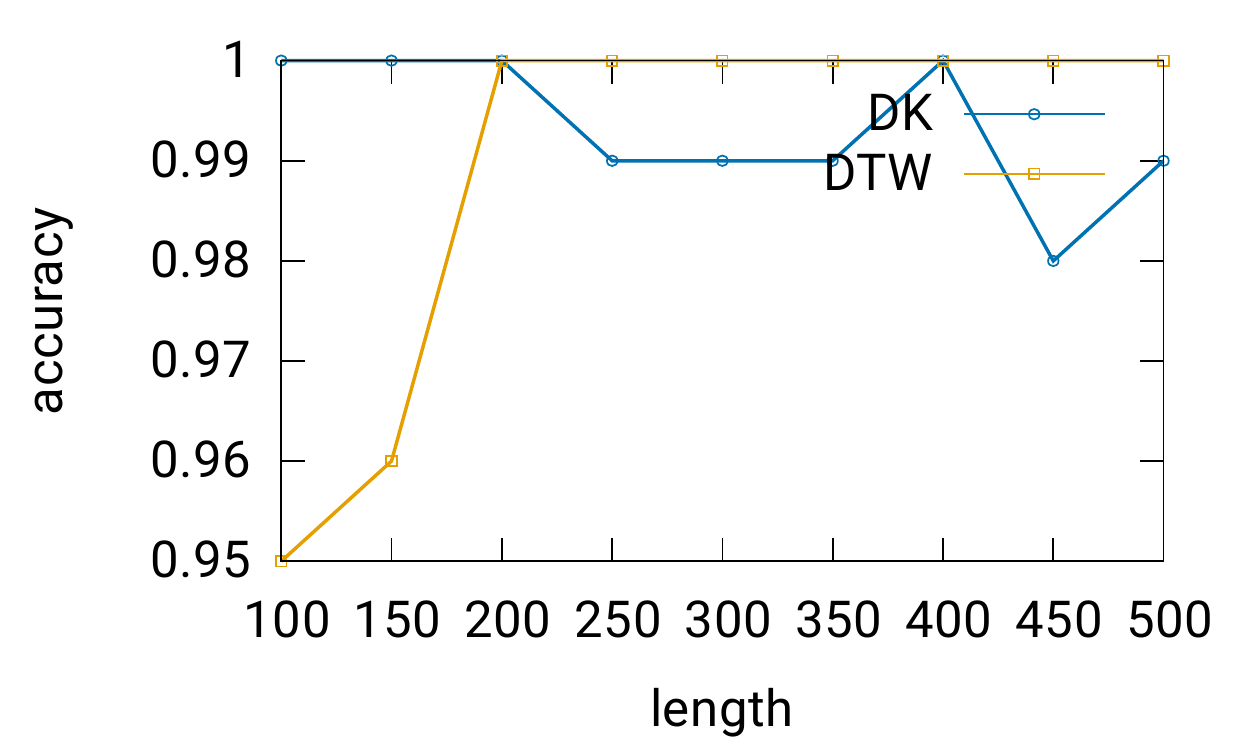}
    \includegraphics[width=.49\linewidth]{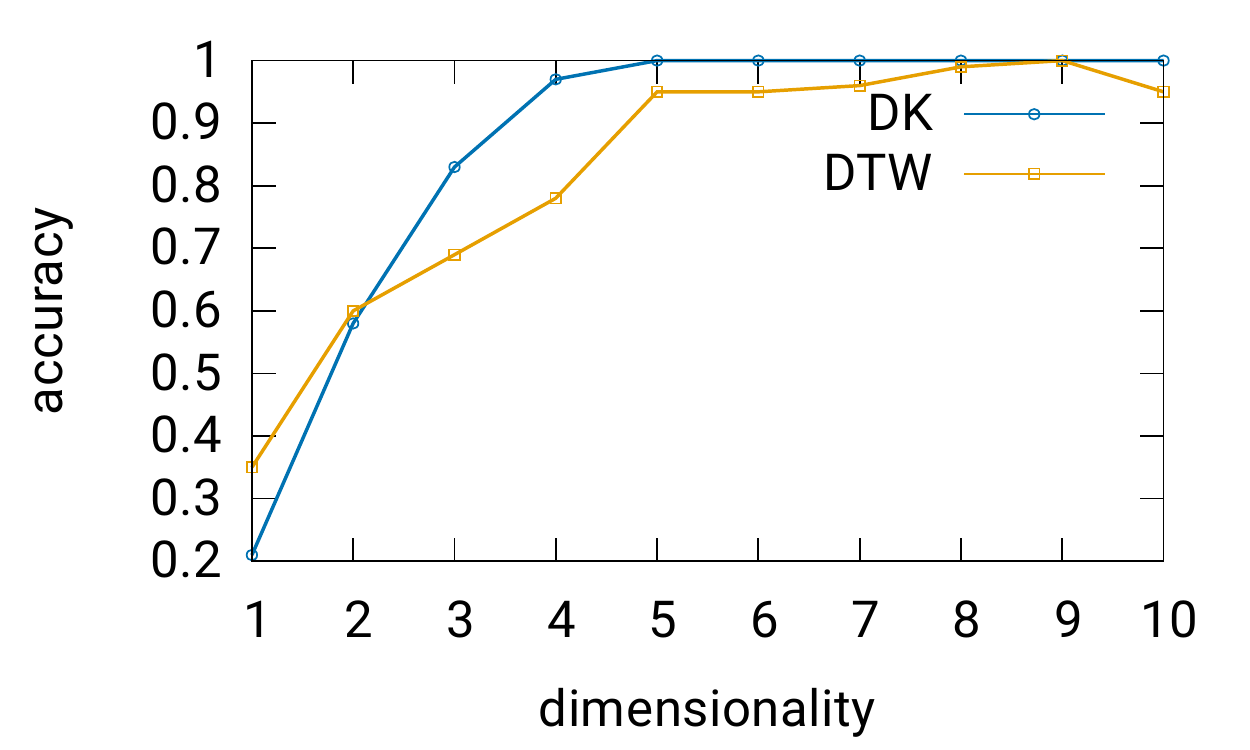}
    \caption{
        Accuracy of \frechet{} and \dtw{} with \lbbox{} on the \ram{} data set. Dimensionality: 10 (left); Length: 100 (right).
    }
    \label{fig:accuracy_ram_len_dim}
\end{figure}

Columns~1,~3~and~4 of Table~\ref{tbl:realdatasets} also prove that there is no clear winner regarding accuracy.

\section{Conclusion}
\label{sec:conclusion}

We introduced \lbbox{} as a canonical extension to Keogh's lower bound \lbkeogh{} for \dtw{} on multi-dimensional time series and proved its correctness.
Not only do lower bounds suffer from the curse of dimensionality even if their tightness remains constant, but we also proved that the tightness of \lbbox{} decreases with growing dimensionality.
On the other hand, we proposed an alternative algorithm for the computation of the long known \frechet{} distance, which is similar to \dtw{}.
Please note that the \frechet{} distance satisfies the triangle inequality and can thus be used in metric indexes \cite{frechet,WDK17}.

In our evaluation, we confirmed our theory that \lbbox{} as extension of \lbkeogh{} suffers from the curse of dimensionality.
We could show that our implementation of \frechet{} outperforms \lbbox{} on multi-dimensional synthesized data sets as well as real world data sets in terms of computation time by more than one order of magnitude.
However, there is no clear winner regarding the accuracy in retrieval tasks.
Hence, we propose to stay with \lbkeogh{} on 1-dimensional time series while choosing the \frechet{} distance on multi-dimensional time series.

\balance

\section{Acknowledgments}
We like to thank Jochen Taeschner for his help on this work.

\bibliographystyle{abbrv}
\bibliography{literature}

\end{document}